\documentclass[letterpaper, 10 pt, conference]{ieeeconf}  

\IEEEoverridecommandlockouts                              

\usepackage{amsmath} 
\usepackage{amssymb}  
\usepackage{graphicx}
\usepackage{xcolor}
\newtheorem{problem}{Problem}
\newtheorem{assumption}{Assumption}
\newtheorem{lemma}{Lemma}
\newtheorem{definition}{Definition}
\newtheorem{remark}{Remark}
\newtheorem{theorem}{Theorem}

\usepackage{cite}

\title{\LARGE \bf
Communication-aware Synthesis of Safe Controllers for Discrete-Time Linear Multi-Agent Systems with Distributed $k$-Hop Observation*
}

\author{Yihan Liu$^{1}$, Teng Yan$^{1}$, Meiqi Tian$^{1}$, Yue Yu$^{1}$ and Bingzhuo Zhong$^{1}$
\thanks{*This work was supported by the National Natural Science Foundation of China (Grant No. 62606433), the Guangdong Basic and Applied Basic Research Foundation (No:2026A1515010222), and the Guangdong Provincial Project (No. 2024QN11X053). (Corresponding Author: Bingzhuo Zhong)}
\thanks{$^{1}$Yihan Liu, Teng Yan, Meiqi Tian, Yue Yu and Bingzhuo Zhong are with the Thrust of Artificial Intelligence, The Hong Kong University of Science and Technology (Guangzhou), Guangzhou 511400, China. (e-mail:\{yliu135, tyan497, mtian837, yyu704\}@connect.hkust-gz.edu.cn, bingzhuoz@hkust-gz.edu.cn).}
}
\begin{document}

\maketitle
\thispagestyle{empty}
\pagestyle{empty}

\begin{abstract}
This paper studies communication-aware safe control for discrete-time
linear multi-agent systems under limited information exchange. The main
challenge lies in the coupling between remote-state estimation and safe
controller design, since estimation errors affect the state evolution
through the controller gains, while the controller design must account
for the resulting observer-induced state perturbations to guarantee
safety. To address this challenge, a distributed input predictor is
incorporated into a $k$-hop state observer to reconstruct unavailable remote states, and the resulting estimation errors are characterized to quantify
the observer-induced state perturbations. An overall robust safety invariant (RSI) set is then
constructed by jointly accounting for individual agent state constraints
and prescribed safety constraints on the relative states of interacting
agents under these perturbations. A linear matrix inequality (LMI)-based
optimization method is developed to jointly synthesize the distributed
observers, local controllers, and the RSI set. A case study illustrates
the effectiveness of the proposed method.
\end{abstract}

\section{Introduction}

Multi-agent systems have been widely used in networked robotics, intelligent transportation, and other cyber-physical applications, where multiple agents are required to cooperate safely. 
In these applications, large-scale multi-agent systems are deployed over communication networks with limited bandwidth, sparse interconnections, and distributed decision-making architectures. Under these settings, the design of control strategies must account for communication constraints, making communication-aware control important in networked systems \cite{Mesbahi2010,Oks2024,Gielis,Ge}.

To cope with the communication limitations on inter-agent information exchange, distributed estimation and observer design have been extensively studied in networked control systems. Representative works include decentralized state estimation for distributed discrete-time systems \cite{Acikmese2014}, pinning-based observer design for networked dynamical systems \cite{EdwardsMenon2016}, distributed output regulation for heterogeneous multi-agent systems \cite{Yang2016}, and consensus-based observer design for multi-agent systems \cite{Zuo2020}. Recently, \cite{GasparriMarino2020} developed a distributed framework for $k$-hop control in large-scale networks, and \cite{Zaccherini2025} further proposed a communication-aware strategy based on $k$-hop distributed observers. These studies show that communication-aware estimation can effectively enlarge the information available for local control. However, existing methods mainly focus on estimation convergence and closed-loop stability, while the impact of observer-induced perturbations on safety guarantees has received less attention.

In parallel, safe controller synthesis has been extensively
studied using invariant-set and robust safety methods
\cite{Blanchini1999,Hsu,Yin}. Recently,
\cite{Zhong2022} developed a direct data-driven method for synthesizing
safety controllers for uncertain linear systems, \cite{Zhong2024}
proposed a set-based approach for formal controller synthesis, and
\cite{Zhong2023} presented a secure-by-construction framework
that enforces safety and security properties.
A robust safe invariant (RSI) synthesis framework was developed in
\cite{lyh} by incorporating communication-induced estimation errors into
controller design. But in multi-agent systems, such estimation errors can also arise from the reconstruction of remote-agent states that are unavailable through direct communication. Under limited inter-agent communication, where each agent can directly exchange information only with its 1-hop neighbors, the remote states required by the local feedback controller are not directly available and must therefore be reconstructed through distributed observation and input prediction. The resulting estimation errors affect the state evolution through the controller gains, thereby coupling the observer and controller designs.

This paper develops a communication-aware framework for safe controller
synthesis of multi-agent systems over a communication graph, where
direct information exchange is restricted to 1-hop neighbors and the
states of agents located 2 to $k$ hops away are reconstructed by
distributed observers supported by an input predictor. The resulting
estimation errors are propagated through the controller gains to
quantify the observer-induced state perturbations, whose effects are
incorporated into an overall RSI set that accounts for both individual
agent states and inter-agent relative-state constraints. An LMI-based
framework is then developed to jointly synthesize the distributed
observers, controllers, and RSI set while satisfying the prescribed
input constraints. The main contributions are summarized as follows.

\begin{enumerate}
    \item We construct a distributed input predictor and a distributed
    $k$-hop state observer for discrete-time linear multi-agent systems, and
    establish the observer-error dynamics and uniform error bounds.

    \item We characterize the observer-induced state perturbation caused by
    remote-state estimation and establish safety conditions for an overall
    RSI set that accounts for both individual agent state constraints and
    inter-agent relative-state constraints.

    \item We develop an LMI-based optimization framework to jointly
    synthesize the observers, controllers, and the overall RSI set while
    accounting for observer-controller coupling and satisfying the input constraints.
\end{enumerate}

\section{PRELIMINARIES AND PROBLEM FORMULATION}

Let $\mathbb R$, $\mathbb R_{>0}$, $\mathbb R_{\ge0}$ and $\mathbb N$ denote the real, positive real, nonnegative real and natural numbers, respectively. The set of $n \times m$ matrices is written as $\mathbb R^{n \times m}$, and
$[A]_{ab}$ denotes the $(a,b)$-th entry of a matrix $A$. For a symmetric matrix $X$, $X\succ0$ ($X\succeq0$) denotes positive (semi-)definiteness, and $I_n$ is the $n$-dimensional identity matrix. The Euclidean and infinity norms are denoted by $\|\cdot\|$ and $\|\cdot\|_\infty$. Moreover, $\otimes$ and $\times$ denote the Kronecker product and the Cartesian product, respectively; $\operatorname{col}(\cdot)$ denotes column-wise stacking; and
$\operatorname{diag}(\cdot)$ denotes a diagonal matrix formed from its
arguments. For a scalar \(a\), \(|a|\) denotes its absolute value, whereas for a finite set \(\mathcal S\), \(|\mathcal S|\) denotes its cardinality.

\subsection{Graph-Theoretic Preliminaries}

Consider a multi-agent system composed of $N$ agents indexed by $
\mathcal V=\{1,2,\dots,N\}.
$
The communication topology is modeled by an undirected graph
$
\mathcal G=(\mathcal V,\mathcal E),
$
where $\mathcal V$ is the node set and
$\mathcal E\subseteq\mathcal V\times\mathcal V$ is the edge set
describing the communication links established among the agents.
An edge $(i,j)\in\mathcal E$ indicates that agents $i$ and $j$ can
directly exchange information.
For each agent $i\in\mathcal V$, its
1-hop neighbor set is defined as
\begin{equation}
    \label{eq:Ni}
    \mathcal N_i:=\{j\in\mathcal V\mid(i,j)\in\mathcal E\}.
\end{equation}
For any two distinct agents $i,j\in\mathcal V$, a path from $i$ to
$j$ is an edge sequence
$\pi_{ij}=(e_1,\ldots,e_\ell)\in\mathcal E^\ell$
connecting them, where $\ell$ is the path length. The graph distance is
$d(i,j):=\min_{\pi_{ij}}|\pi_{ij}|$. Given an integer $k\ge2$, the set of
remote agents whose graph distance from agent $i$ is between $2$ and
$k$ is defined as
\begin{equation}
\mathcal N_i^{k\text{-hop}}
:=
\{j\in\mathcal V\mid 2\le d(i,j)\le k\}.
\label{eq:khop_neighbors}
\end{equation}
Let $
\eta_i:=|\mathcal N_i^{k\text{-hop}}|
$
denote the number of such agents. Agents in
$\mathcal N_i^{k\text{-hop}}$ are referred to as the remote agents of
agent $i$.

To specify the information available for distributed estimation and
control over the communication graph, we impose the following assumption.

\begin{assumption}
At each time instant, each agent $i\in\mathcal V$ has direct access to
its own state and input. Moreover,
it can transmit the locally available state and input to its 1-hop neighbors. The
communication delay within each sampling interval is assumed to be
negligible.
\label{as:local_information}
\end{assumption}

Assumption~\ref{as:local_information} describes synchronous local
communication over $\mathcal G$: each agent measures its own variables
and exchanges locally available information only with adjacent agents.
It is reasonable when communication is sufficiently fast relative to
the sampling period~\cite{Zaccherini2025}.

\subsection{Interconnected Multi-Agent System}
\label{2.B}

Consider the interconnected discrete-time linear multi-agent system associated with a graph $\mathcal G$. Define the overall state and
control input as
$
x(t):=
\begin{bmatrix}
x_1^\top(t)&\cdots&x_N^\top(t)
\end{bmatrix}^\top
\in\mathbb R^{Nn_x}
$
and
$
u(t):=
\begin{bmatrix}
u_1^\top(t)&\cdots&u_N^\top(t)
\end{bmatrix}^\top
\in\mathbb R^{Nn_u},
$
where, for each $i\in\mathcal V$, $x_i(t)\in\mathbb R^{n_x}$ and
$u_i(t)\in\mathbb R^{n_u}$ denote the state and control input
associated with agent $i$, respectively.
The interconnected multi-agent system evolves as
\begin{equation}
x(t+1)=Ax(t)+Bu(t),
\quad \forall t\in\mathbb N,
\label{eq:overall_dynamics}
\end{equation}
where
$
A:=\operatorname{blkdiag}(A_1,\ldots,A_N)
\in\mathbb R^{Nn_x\times Nn_x},
$
and
$
B:=\operatorname{blkdiag}(B_1,\ldots,B_N)
\in\mathbb R^{Nn_x\times Nn_u},
$
with $A_i\in\mathbb R^{n_x\times n_x}$ and
$B_i\in\mathbb R^{n_x\times n_u}$ denoting the system and input
matrices associated with agent $i$, respectively.
The admissible input set of the interconnected system is defined as
\begin{equation}
\mathcal U
:=
\left\{
u\in\mathbb R^{Nn_u}
\middle|
\|u_i\|_\infty\le u_{\max}^i,\ 
\forall i\in\mathcal V
\right\},
\label{eq:overall_input_set}
\end{equation}
where $u_{\max}^i\in\mathbb R_{>0}$ denotes the prescribed input bound
associated with agent $i$, and $u$ is bounded as in \eqref{eq:overall_input_set}. 
Thus, the control input is required to satisfy $u(t)\in\mathcal U$ for all $t\in\mathbb N$. 
For each agent $i\in\mathcal V$, define
$
\mathcal R_i
:=
\mathcal N_i\cup\mathcal N_i^{k\text{-hop}},
$
which contains its 1-hop neighbors and the remote agents located
2 to $k$ hops away. The safe set of the
interconnected system is defined as
\begin{equation}
\mathcal X:=\left\{
x\in\mathbb R^{Nn_x}
\middle|
\begin{array}{@{}ll@{}}
C_i x_i\le\mathbf 1,
&\forall i\in\mathcal V,
\\
|d_i(x_i-x_p)|\le b_p^i,
&\forall i\in\mathcal V,\ p\in\mathcal R_i
\end{array}
\right\},
\label{eq:overall_safe_set1}
\end{equation}
where
$C_i\in\mathbb R^{n_{c_i}\times n_x}$ specifies the state constraints
associated with agent $i$,
$\mathbf 1\in\mathbb R^{n_{c_i}}$ is the vector of ones,
$d_i\in\mathbb R^{1\times n_x}$ specifies the direction of the
pairwise relative-state constraint, and
$b_p^i\in\mathbb R_{>0}$ denotes the corresponding relative-state safety bound. The safety requirement throughout
this paper is therefore $x(t)\in\mathcal X$ for all $t\in\mathbb N$.

The dynamics of subsystem $i\in\mathcal V$ are given by
\begin{equation}
x_i(t+1)=A_i x_i(t)+B_i u_i(t),
\quad \forall i\in\mathcal V,\quad \forall t\in\mathbb N.
\label{eq:agent_dynamics}
\end{equation}
For each subsystem $i\in\mathcal V$, define its state safe set as
\begin{equation}
\mathcal X^i
:=
\left\{
x_i\in\mathbb R^{n_x}
\mid
C_i x_i\le\mathbf 1
\right\}.
\label{eq:safe_set}
\end{equation}
Moreover, for each $p\in\mathcal R_i$, define the relative state as
$
y_p^i(t):=x_i(t)-x_p(t),
$
and the corresponding safe set as
\begin{equation}
\mathcal S_p^i
:=
\left\{
y_p^i\in\mathbb R^{n_x}
\mid
|d_i y_p^i|\le b_p^i
\right\}.
\label{eq:relative_safe_set}
\end{equation}
Accordingly, the safety requirement
$x(t)\in\mathcal X$ for all $t\in\mathbb N$ is equivalent to requiring
$x_i(t)\in\mathcal X^i$ and
$y_p^i(t)\in\mathcal S_p^i$
for all $i\in\mathcal V$, $p\in\mathcal R_i$, and $t\in\mathbb N$.

\subsection{Observer System}

For each agent $i\in\mathcal V$, write its remote-agent set as
$
\mathcal N_i^{k\text{-hop}}
=
\{j_1,\ldots,j_{\eta_i}\},
$
where $\mathcal N_i^{k\text{-hop}}$ is defined in
\eqref{eq:khop_neighbors}.
The states of these remote agents are required by the local
controller of agent $i$ to incorporate $k$-hop state information.
Define the corresponding stacked remote-state vector as
$
x_i^o(t):=
\begin{bmatrix}
x_{j_1}^\top(t)&\cdots&x_{j_{\eta_i}}^\top(t)
\end{bmatrix}^\top
\in\mathbb R^{\eta_i n_x}.$
Since the states in $x_i^o(t)$ are not directly available to agent $i$
through 1-hop communication, a distributed $k$-hop observer is employed
to reconstruct them. For each
$j\in\mathcal N_i^{k\text{-hop}}$, let
$\hat x_j^i(t)\in\mathbb R^{n_x}$ denote the estimate of $x_j(t)$
maintained by agent $i$. The corresponding stacked estimate is defined as
\begin{equation}
\hat x_i^o(t):=
\begin{bmatrix}
(\hat x_{j_1}^{i}(t))^\top&\cdots&
(\hat x_{j_{\eta_i}}^{i}(t))^\top
\end{bmatrix}^\top
\in\mathbb R^{\eta_i n_x}.
\label{eq:estimated_remote_state}
\end{equation}
Similarly, define
$
\hat u_i^o(t)
=
[(\hat u_{j_1}^i(t))^\top\ \cdots\
(\hat u_{j_{\eta_i}}^i(t))^\top]^\top
\in\mathbb R^{\eta_i n_u},
$
where $\hat u_j^i(t)\in\mathbb R^{n_u}$ denotes the estimate of
$u_j(t)$ maintained by agent $i$. 
The distributed $k$-hop observer is represented as
\begin{equation}
\hat x_i^o(t+1)
=
\Phi_i^o
\left(
\hat x_i^o(t),
\hat u_i^o(t),
\mathcal I_i(t)
\right),
\label{eq:abstract_observer}
\end{equation}
where $\mathcal I_i(t)$ denotes the locally available information of
agent $i$ under Assumption~\ref{as:local_information}, and
$\Phi_i^o(\cdot)$ denotes the distributed observer update map.

\begin{figure}[!t]
    \centering   
    \includegraphics[width=0.9\linewidth]{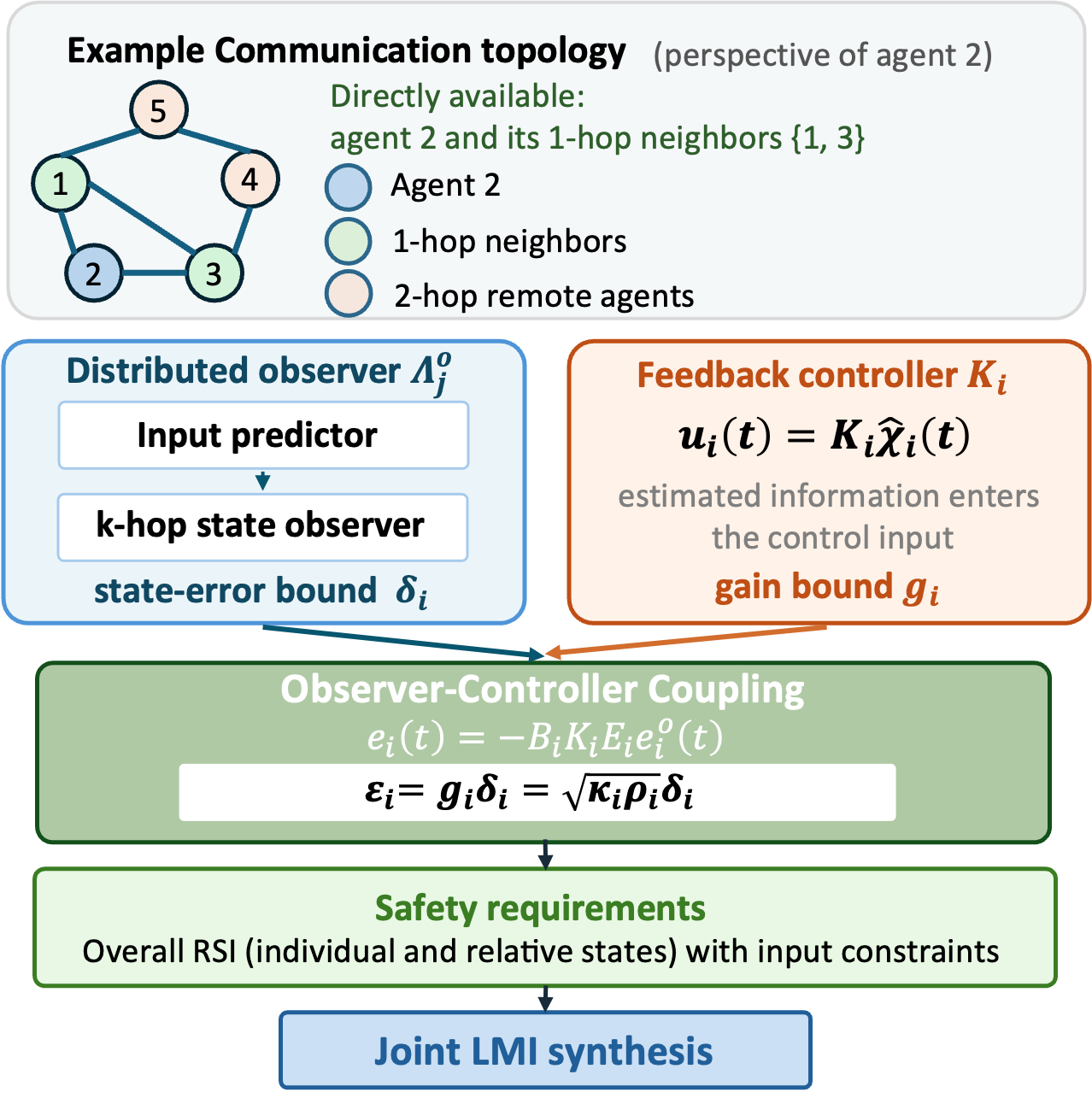}
    \caption{Coupled observer–controller synthesis for communication-aware safe control with distributed \(k\)-hop observation.}   
    \label{fig:agent2_framework}
\end{figure}

\subsection{Control System}

For each agent $i\in\mathcal V$, a distributed state-feedback controller
is implemented using its locally available and estimated $k$-hop state
information. Define the full-information vector as
$
\chi_i(t):=
\begin{bmatrix}
x_i^\top(t)&
x_{\mathcal N_i}^\top(t)&
x_i^{o\top}(t)
\end{bmatrix}^\top
\in\mathbb R^{n_{\chi_i}},
$
where
$
n_{\chi_i}:=(1+|\mathcal N_i|+\eta_i)n_x,
$
and
$x_{\mathcal N_i}(t)\in\mathbb R^{|\mathcal N_i|n_x}$
stacks the states associated with the 1-hop neighbor set
$\mathcal N_i$ defined in \eqref{eq:Ni}.
Since the states in $x_i^o(t)$ are not directly available to agent $i$,
their estimates in \eqref{eq:estimated_remote_state} are used to construct
the information vector available to the local controller as
$
\hat\chi_i(t):=
\begin{bmatrix}
x_i^\top(t)&
x_{\mathcal N_i}^\top(t)&
\hat x_i^{o\top}(t)
\end{bmatrix}^\top
\in\mathbb R^{n_{\chi_i}}.$ The state-feedback controller of agent $i$ is then given by
\begin{equation}
u_i(t):=K_i\hat\chi_i(t),
\quad
\forall i\in\mathcal V,\quad
\forall t\in\mathbb N,
\label{eq:e_feedback_law}
\end{equation}
where
$
K_i\in\mathbb R^{n_u\times n_{\chi_i}}
$
is the state-feedback gain to be jointly synthesized in Section~\ref{section4}.

\subsection{Problem Formulation}

The observer and controller designs are coupled. The local
controller relies on the remote-state estimates $\hat x_i^o(t)$ through
$\hat\chi_i(t)$, while the corresponding estimation errors affect the
closed-loop state evolution through the feedback gain $K_i$. Therefore,
the observers and controllers must be jointly considered
in the safety synthesis. Fig.~\ref{fig:agent2_framework} shows the proposed framework for
the joint observer-controller synthesis. 

\begin{problem}
\label{problem}
Consider the interconnected multi-agent system in
\eqref{eq:overall_dynamics} under Assumption~\ref{as:local_information},
with the admissible input set $\mathcal U$ defined in
\eqref{eq:overall_input_set} and the safe set $\mathcal X$
defined in \eqref{eq:overall_safe_set1}. Jointly synthesize the distributed
$k$-hop observers in \eqref{eq:abstract_observer}, the distributed
state-feedback controllers in \eqref{eq:e_feedback_law}, and a
safe invariant set $\mathcal X_{\mathrm{RSI}}\subseteq\mathcal X$ such
that, for every $x(0)\in\mathcal X_{\mathrm{RSI}}$, the resulting
closed-loop system satisfies $x(t)\in\mathcal X_{\mathrm{RSI}}$ and
$u(t)\in\mathcal U$ for all $t\in\mathbb N$.
\end{problem}

\section{DISTRIBUTED $k$-HOP OBSERVER AND ERROR ANALYSIS}
\label{section3}

Building on the subsystem dynamics in~\eqref{eq:agent_dynamics}, this section constructs the distributed $k$-hop observer represented
in~\eqref{eq:abstract_observer}. 
The construction combines a distributed
input predictor with a distributed state observer to estimate remote
information. We then characterize the
evolution of the estimation errors and the resulting observer-induced
state perturbation.

\subsection{Observer Information Structure}


Before constructing the distributed observer, it is necessary to
identify which agents estimate each remote state and how correction
information is exchanged among them. The following
definition introduces the
corresponding observer information structure.

\begin{definition}\label{def:observer_information_structure}
For each agent $j\in\mathcal V$ with $\eta_j>0$, define
$
\mathcal O_j
:=
\{r\in\mathcal V\mid j\in\mathcal N_r^{k\text{-hop}}\}.
$
The set $\mathcal O_j$ collects all agents that maintain an estimate of
$x_j(t)$. Since $\mathcal G$ is undirected,
$
\mathcal O_j=\mathcal N_j^{k\text{-hop}},
$
and hence $|\mathcal O_j|=\eta_j$. Let
$
\mathcal O_j=\{r_1,\ldots,r_{\eta_j}\}.
$
Let $D_j\in\mathbb R^{\eta_j\times\eta_j}$ denote the Laplacian matrix
of the communication links among the agents in
$\mathcal O_j$.
$D_j$ characterizes the estimate-consensus coupling among
the agents maintaining estimates of $x_j(t)$. Define
$
H_j:=
\operatorname{diag}\bigl(
|\mathcal N_{r_1}\cap\mathcal N_j|,
\ldots,
|\mathcal N_{r_{\eta_j}}\cap\mathcal N_j|
\bigr)
\in\mathbb R^{\eta_j\times\eta_j},
$
where $H_j$ characterizes the true-state information available through
common 1-hop neighbors. Define
\begin{equation}
M_j:=D_j+H_j
\in\mathbb R^{\eta_j\times\eta_j},
\label{eq:Mkc}
\end{equation}
which combines the estimate-consensus coupling and the true-state
pinning effect. Its lifted form is defined as
$
\bar M_j:=M_j\otimes I_{n_x}
\in\mathbb R^{\eta_j n_x\times\eta_j n_x}.
$
\end{definition}

Definition~\ref{def:observer_information_structure} describes the
information exchange among the agents estimating $x_j(t)$ and their
access to its true value. To ensure that this information reaches every
estimator, we impose the following connectivity and pinning condition.

\begin{assumption}
\label{ass:graph_condition}
Based on Definition~\ref{def:observer_information_structure}, for each agent $j\in\mathcal V$ with $\eta_j>0$, the following graph
conditions hold:

1. the communication graph $\mathcal G$ is connected, i.e., there exists a path between every pair of distinct agents in $\mathcal V$;

2. for every connected component of the subgraph induced by
$\mathcal O_j$, there exists at least one node whose corresponding
diagonal entry in $H_j$ is strictly positive.
\end{assumption}

Assumption~\ref{ass:graph_condition} is satisfied by ensuring
that each connected component contains at least one agent with access
to true-state correction information. It then allows this information
to reach every estimator through consensus interactions. The following
lemma expresses this property through the positive definiteness of
$M_j$.

\begin{lemma}
\label{M_kc}
Based on Definition~\ref{def:observer_information_structure} and
Assumption~\ref{ass:graph_condition}, for each agent
$j\in\mathcal V$ with $\eta_j>0$, the matrix $M_j$ defined in
\eqref{eq:Mkc} is positive definite. Consequently,
$
\bar M_j\succ0.
$
\end{lemma}


The proof of Lemma~\ref{M_kc} is presented in the Appendix. Assumption~\ref{ass:graph_condition}
ensures that true-state correction information can reach every estimator
through consensus interactions, while Lemma~\ref{M_kc} confirms this
property through the positive definiteness of $M_j$. Together, they
establish the graph-theoretic foundation for the distributed observer.

\subsection{Distributed Input Prediction}

Remote-state estimation requires knowledge of control
inputs, which are not directly available beyond 1-hop communication.
Therefore, a distributed input predictor is introduced to reconstruct this information over the observer network. The following
lemma presents the input predictor.

\begin{lemma}
\label{lem:input_predictor}
Under Definition~\ref{def:observer_information_structure} and
Lemma~\ref{M_kc}, for each $j\in\mathcal V$ with $\eta_j>0$, let
$\mathcal O_j:=\mathcal N_j^{k\text{-hop}}
=\{r_1,\ldots,r_{\eta_j}\}$ denote the set of agents that
maintain an estimate of $u_j(t)$. The distributed input predictor is
\begin{equation}
\mathbf{\hat u}_j^{o,+}(t)
=
\mathbf{\hat u}_j^o(t)+s_j\xi_{u,j}^o(t),\quad 
\mathbf{\hat u}_j^o(t+1)
=
\mathbf{\hat u}_j^{o,+}(t),
\quad t\in\mathbb N,
\label{eq:corrected_input_estimate}
\end{equation}
where $\mathbf{\hat u}_j^{o,+}(t)$ is the corrected stacked estimate of
$u_j(t)$; $\mathbf{\hat u}_j^o(t):=
\operatorname{col}_{r\in\mathcal O_j}(\hat u_j^r(t))$ is the stacked
input estimate before correction, where each $\hat u_j^r(t)$ is locally estimated by agent
$r\in\mathcal O_j$;
$\xi_{u,j}^o(t):=
\operatorname{col}_{r\in\mathcal O_j}(\xi_{u,j}^r(t))$ is the stacked correction; and
$s_j:=2/(\lambda_{\min}(M_j)+\lambda_{\max}(M_j))$ is the correction
step size, where $M_j$ is defined in \eqref{eq:Mkc} and is positive
definite by Lemma~\ref{M_kc}. For each $r\in\mathcal O_j$, the input
correction is
\begin{equation}
\xi_{u,j}^{r}(t)
:=
\sum_{\ell\in\mathcal N_{r}\cap\mathcal O_j}
\bigl(\hat u_j^\ell(t)-\hat u_j^{r}(t)\bigr)
+
|\mathcal N_r\cap\mathcal N_j|
\bigl(u_j(t)-\hat u_j^{r}(t)\bigr),
\label{eq:input_observer_correction}
\end{equation}
where the neighbor sets $\mathcal N_r$ and $\mathcal N_j$ are defined
in \eqref{eq:Ni}.
\end{lemma}

The proof of Lemma~\ref{lem:input_predictor} is presented in
the Appendix. Lemma~\ref{lem:input_predictor} provides the distributed mechanism for
reconstructing unavailable inputs. To quantify its prediction accuracy, the following theorem establishes a
uniform bound on the input-estimation error.

\begin{theorem}
\label{thm:input_error_bound}
Under Definition~\ref{def:observer_information_structure},
Lemma~\ref{M_kc}, and Lemma~\ref{lem:input_predictor}, consider any
$j\in\mathcal V$ with $\eta_j>0$. Suppose that there exist scalars
$l_j\in\mathbb R_{\ge0}$ and $\bar u_j^o\in\mathbb R_{\ge0}$ such that
$\|u_j(t+1)-u_j(t)\|\le l_j$ for all $t\in\mathbb N$ and
$\|\mathbf e_{u,j}^{o,+}(0)\|\le\bar u_j^o$. Then
\begin{equation}
\|\mathbf e_{u,j}^{o,+}(t)\|
\le
\nu_j,
\quad
\nu_j:=
\max\left\{
\bar u_j^o,
\frac{\sigma_j\sqrt{\eta_j}l_j}{1-\sigma_j}
\right\},
\quad
\forall t\in\mathbb N,
\label{eq:input_error_bound}
\end{equation}
where
$\mathbf e_{u,j}^{o,+}(t):=
\mathbf 1_{\eta_j}\otimes u_j(t)
-\mathbf{\hat u}_j^{o,+}(t)
\in\mathbb R^{\eta_jn_u}$
is the corrected stacked input-estimation error;
$l_j$ is an upper bound on the one-step variation of $u_j(t)$;
$\bar u_j^o$ is an upper bound on the initial corrected error;
$\nu_j\in\mathbb R_{\ge0}$ is the resulting uniform error bound; and
$\sigma_j:=\|I_{\eta_j}-s_jM_j\|\in[0,1)$ is the contraction factor
associated with $s_j$ specified in
Lemma~\ref{lem:input_predictor} and $M_j$
defined in \eqref{eq:Mkc}.
\end{theorem}


The proof of Theorem~\ref{thm:input_error_bound} is presented in the
Appendix. The theorem shows that the distributed input predictor tracks
time-varying inputs with a uniformly bounded error under bounded input
variations. Thus, this subsection provides both the input-prediction
mechanism and its error guarantee for the state observer.

\subsection{Distributed State Observation and Error Bound}

Using the corrected input estimates obtained above, we now construct a
distributed state observer to reconstruct the unavailable remote states.
The following definition specifies its structure.

\begin{definition}
\label{def:distributed_state_observer}
Under Definition~\ref{def:observer_information_structure} and
Lemma~\ref{lem:input_predictor}, consider any
$j\in\mathcal V$ with $\eta_j>0$, where
$\mathcal O_j=\{r_1,\ldots,r_{\eta_j}\}$ is the set of agents that
maintain an estimate of $x_j(t)$. The distributed state observer is
\begin{equation}
\mathbf{\hat x}_j^o(t+1)
=
\mathbf A_j^o\mathbf{\hat x}_j^o(t)
+
\mathbf B_j^o\mathbf{\hat u}_j^{o,+}(t)
+
\Lambda_j^o\xi_j^o(t),
\quad t\in\mathbb N,
\label{eq:observer}
\end{equation}
where
$\mathbf{\hat x}_j^o(t):=
\operatorname{col}_{r\in\mathcal O_j}(\hat x_j^r(t))
\in\mathbb R^{\eta_jn_x}$
is the stacked state estimate, and
$\hat x_j^r(t)\in\mathbb R^{n_x}$ is the estimate of $x_j(t)$
maintained by agent $r\in\mathcal O_j$;
$\xi_j^o(t):=
\operatorname{col}_{r\in\mathcal O_j}(\xi_j^r(t))
\in\mathbb R^{\eta_jn_x}$
is the stacked state-correction term;
$\mathbf A_j^o:=I_{\eta_j}\otimes A_j
\in\mathbb R^{\eta_jn_x\times\eta_jn_x}$
and
$\mathbf B_j^o:=I_{\eta_j}\otimes B_j
\in\mathbb R^{\eta_jn_x\times\eta_jn_u}$
are the lifted subsystem matrices; and
$\mathbf{\hat u}_j^{o,+}(t)\in\mathbb R^{\eta_jn_u}$
is the corrected stacked input estimate in
\eqref{eq:corrected_input_estimate}.
The stacked observer gain is
\begin{equation}
\Lambda_j^o
:=
\operatorname{blkdiag}
\left(
\Lambda_j^{r_1},
\ldots,
\Lambda_j^{r_{\eta_j}}
\right)
\in\mathbb R^{\eta_jn_x\times\eta_jn_x},
\label{eq:stacked_observer_gain}
\end{equation}
where $\Lambda_j^r\in\mathbb R^{n_x\times n_x}$ is the local observer
gain used by agent $r$. For each $r\in\mathcal O_j$, the local
state-correction vector is
\begin{equation}
\xi_j^{r}(t)
:=
\sum_{\ell\in\mathcal N_{r}\cap\mathcal O_j}
\bigl(\hat x_j^\ell(t)-\hat x_j^{r}(t)\bigr)
+
|\mathcal N_r\cap\mathcal N_j|
\bigl(x_j(t)-\hat x_j^{r}(t)\bigr),
\label{eq:implementable_innovation}
\end{equation}
where $\mathcal N_r$ and $\mathcal N_j$ are defined in
\eqref{eq:Ni}.
\end{definition}

Based on the distributed state observer above, the following lemma
characterizes how the state-estimation error evolves under the
corrected input estimates.

\begin{lemma}
\label{lem:observer_error_dynamics}
Under the subsystem dynamics in \eqref{eq:agent_dynamics} and
Definition~\ref{def:distributed_state_observer}, the observer error
satisfies
\begin{equation}
\mathbf e_j^o(t+1)
=
\Psi_j\mathbf e_j^o(t)
+
\mathbf B_j^o\mathbf e_{u,j}^{o,+}(t),
\qquad
\forall t\in\mathbb N,
\label{eq:error_dynamics}
\end{equation}
where
$\mathbf e_j^o(t):=
\operatorname{col}_{r\in\mathcal O_j}
(x_j(t)-\hat x_j^r(t))
\in\mathbb R^{\eta_jn_x}$
is the stacked state-estimation error;
$\Psi_j:=\mathbf A_j^o-\Lambda_j^o\bar M_j
\in\mathbb R^{\eta_jn_x\times\eta_jn_x}$
is the homogeneous observer-error transition matrix, with
$\bar M_j$ defined in
Definition~\ref{def:observer_information_structure}; and
$\mathbf e_{u,j}^{o,+}(t)\in\mathbb R^{\eta_jn_u}$
is the corrected stacked input-estimation error defined in
Theorem~\ref{thm:input_error_bound}.
\end{lemma}

The proof is presented in
the Appendix. Lemma~\ref{lem:observer_error_dynamics} shows that the state-estimation
error is driven by the input-estimation error that is uniformly bounded
by Theorem~\ref{thm:input_error_bound}. Therefore, a uniform
state-estimation error bound can be established if the homogeneous
observer-error dynamics are contractive, motivating the following
assumption.

\begin{assumption}\label{ass:observer_stability}
Based on Lemma~\ref{lem:observer_error_dynamics}, for each
agent $j\in\mathcal V$ with $\eta_j>0$, there exist scalars
$q_j^o\in(0,1)$ and $\gamma_j\in[1,\infty)$, an observer gain
$
\Lambda_j^o\in
\mathbb R^{\eta_j n_x\times\eta_j n_x}
$
of the form \eqref{eq:stacked_observer_gain}, and a symmetric matrix
$
P_j\in\mathbb R^{\eta_j n_x\times\eta_j n_x}, P_j\succ0,
$
such that
\begin{equation}
\Psi_j^\top P_j\Psi_j
\preceq q_j^oP_j,
\qquad
I_{\eta_j n_x}
\preceq P_j
\preceq
\gamma_j I_{\eta_j n_x}.
\label{eq:observer_P_bounds}
\end{equation}
\end{assumption}

Assumption~\ref{ass:observer_stability} requires the homogeneous
observer-error dynamics in \eqref{eq:error_dynamics} to be contractive. It can be enforced through the LMIs developed in Section~\ref{section4}. 
Together with the bounded input-estimation error established in
Theorem~\ref{thm:input_error_bound}, this property allows us to characterize
the evolution of the state-estimation error and derive a uniform bound, as stated in the following theorem.


\begin{theorem}
\label{thm:observer_convergence_bound}
Under Definition~\ref{def:distributed_state_observer},
Lemma~\ref{lem:observer_error_dynamics},
Assumption~\ref{ass:observer_stability}, and
Theorem~\ref{thm:input_error_bound},  consider any
$j\in\mathcal V$ with $\eta_j>0$. Suppose that there exists a scalar
$\bar e_{j,0}^o\in\mathbb R_{\ge0}$ such that
$\|\mathbf e_j^o(0)\|\le\bar e_{j,0}^o$. Then
\begin{equation}
\|\mathbf e_j^o(t)\|
\le
\bar e_j^o,
\quad
\bar e_j^o:=
\sqrt{\gamma_j}
\left(
\bar e_{j,0}^o
+
\frac{\|\mathbf B_j^o\|\nu_j}{1-\sqrt{q_j^o}}
\right),
\quad
\forall t\in\mathbb N,
\label{eq:uniform_state_observer_error}
\end{equation}
where $\mathbf e_j^o(t)\in\mathbb R^{\eta_jn_x}$ is the stacked
state-estimation error defined in
Lemma~\ref{lem:observer_error_dynamics};
$\bar e_{j,0}^o$ is an upper bound on its initial norm;
$q_j^o\in(0,1)$ and $\gamma_j\in[1,\infty)$ are the parameters in Assumption~\ref{ass:observer_stability};
$\nu_j\in\mathbb R_{\ge0}$ is the uniform input-estimation error bound
established in~\eqref{eq:input_error_bound}; and
$\bar e_j^o\in\mathbb R_{\ge0}$ is the resulting uniform
state-estimation error bound.
\end{theorem}

The proof of Theorem~\ref{thm:observer_convergence_bound} is presented
in the Appendix.
Theorem~\ref{thm:observer_convergence_bound} establishes a uniform bound
on the remote-state estimation errors. These errors propagate through the controller gains
and appear as additive perturbations in the closed-loop dynamics. The
following theorem quantifies this effect and provides a uniform bound on the
resulting observer-induced state perturbation.

\begin{theorem}
\label{thm:ei_bound}
Under the subsystem dynamics in \eqref{eq:agent_dynamics}, the
feedback law in \eqref{eq:e_feedback_law}, and
Theorem~\ref{thm:observer_convergence_bound}, consider any
$i\in\mathcal V$ with $\eta_i>0$. Let
$e_i^o(t):=x_i^o(t)-\hat x_i^o(t)
\in\mathbb R^{\eta_i n_x}$ denote the stacked remote-state estimation
error, and let
$E_i\in\mathbb R^{n_{\chi_i}\times\eta_i n_x}$ denote the selector
matrix satisfying
$\chi_i(t)-\hat\chi_i(t)=E_ie_i^o(t)$.
If there exists a scalar
$g_i\in\mathbb R_{\ge0}$ such that
$
(B_iK_iE_i)(B_iK_iE_i)^\top\preceq g_i^2I_{n_x},
$
then the observer-induced state perturbation satisfies
\begin{equation}
\|e_i(t)\|\le\varepsilon_i,
\qquad
\varepsilon_i:=g_i\delta_i,
\qquad
\forall t\in\mathbb N,
\label{eq:state_perturbation_bound}
\end{equation}
where
\begin{equation}
\delta_i
:=
\sqrt{\sum_{j\in\mathcal N_i^{k\text{-hop}}}(\bar e_j^o)^2},
\label{eq:delta_i}
\end{equation}
is the aggregate remote-state observer-error bound constructed from
the individual bounds $\bar e_j^o$ over
$\mathcal N_i^{k\text{-hop}}$; $g_i$ bounds the induced gain from
$e_i^o(t)$ to $e_i(t)$; and
$\varepsilon_i\in\mathbb R_{\ge0}$ is the resulting admissible
perturbation radius. Here, $\mathcal N_i^{k\text{-hop}}$ is defined in
\eqref{eq:khop_neighbors}, $\bar e_j^o$ is established in~\eqref{eq:uniform_state_observer_error}, and
$e_i(t):=-B_iK_iE_ie_i^o(t)\in\mathbb R^{n_x}$ is the observer-induced
state perturbation caused by using
the estimated information $\hat\chi_i(t)$ instead of the full
information $\chi_i(t)$.
\end{theorem}

The proof of Theorem~\ref{thm:ei_bound} is presented in the Appendix.
Theorem~\ref{thm:ei_bound} shows that using estimated information
produces an additive closed-loop state perturbation relative to the ideal full-information closed-loop dynamics, whose magnitude is bounded in terms of the observer errors and controller gain.
Overall, this section constructs the graph-based $k$-hop observation
scheme, establishes uniform estimation error bounds,
and translates them into a perturbation bound for the safety controller design.

\section{COMMUNICATION-AWARE SAFE CONTROLLER DESIGN}
\label{section4}

Building on Theorem~\ref{thm:ei_bound}, this section constructs an
overall RSI set that accounts for both individual agent states and
inter-agent relative-state constraints under observer-induced
perturbations. We then formulate an LMI-based problem to jointly
synthesize the distributed observers, controllers, and RSI set subject
to input constraints.

\subsection{Overall RSI Set}

To formulate the safety requirements for controller synthesis, we first
incorporate the observer-induced perturbations into the safe set
description. The following definition formalizes the robust
invariance requirement.

\begin{definition}
\label{def:epsilon_rsi}
Under the overall system dynamics in \eqref{eq:overall_dynamics}, the
feedback law in \eqref{eq:e_feedback_law}, and
Theorem~\ref{thm:ei_bound}, for each $i\in\mathcal V$, define the
local ellipsoid
\begin{equation}
\mathcal X_{\mathrm{safe}}^i
:=
\left\{
x_i\in\mathbb R^{n_x}
\mid
x_i^\top Q_i x_i\le1
\right\},
Q_i\in\mathbb R^{n_x\times n_x},
Q_i\succ0,
\label{eq:epsilon_RSI_set}
\end{equation}
where $x_i\in\mathbb R^{n_x}$ denotes the state of agent $i$ introduced
in \eqref{eq:agent_dynamics}, and
$\mathcal X_{\mathrm{safe}}^i\subseteq\mathcal X^i$. The overall candidate safe set is
\begin{equation}
\mathcal X_{\mathrm{RSI}}
:=
\left\{
x\in\mathbb R^{Nn_x}
\middle|
\begin{array}{ll}
x_i\in\mathcal X_{\mathrm{safe}}^i,
&
\forall i\in\mathcal V,
\\
y_p^i\in\mathcal S_p^i,
&
\forall i\in\mathcal V,
\forall p\in\mathcal R_i
\end{array}
\right\},
\label{eq:global_RSI_candidate}
\end{equation}
where $\mathcal X_{\mathrm{RSI}}\subseteq\mathcal X$;
$x=\operatorname{col}_{i\in\mathcal V}(x_i)$ denotes the overall state
corresponding to $x(t)$ in \eqref{eq:overall_dynamics};
$\mathcal R_i=\mathcal N_i\cup\mathcal N_i^{k\text{-hop}}$ is defined
in Section~\ref{2.B}; and $y_p^i:=x_i-x_p$ and
$\mathcal S_p^i$ are the relative state and its admissible set defined
in \eqref{eq:relative_safe_set}, respectively.
Define the local and overall admissible perturbation sets as
$
\Delta(\varepsilon_i)
:=
\left\{
e_i\in\mathbb R^{n_x}
\mid
\|e_i\|\le\varepsilon_i
\right\},
\Delta(\boldsymbol\varepsilon)
:=
\prod_{i\in\mathcal V}\Delta(\varepsilon_i),
$
where
$\boldsymbol\varepsilon
:=\operatorname{col}_{i\in\mathcal V}(\varepsilon_i)$ and
$\varepsilon_i\in\mathbb R_{\ge0}$ is 
established in Theorem~\ref{thm:ei_bound}.
The closed-loop dynamics are
$x(t+1)=
Ax(t)
+
B\operatorname{col}_{i\in\mathcal V}
\bigl(K_i\chi_i(t)\bigr)
+
e(t),
e(t):=\operatorname{col}_{i\in\mathcal V}\bigl(e_i(t)\bigr)
\in\Delta(\boldsymbol\varepsilon).$
The set $\mathcal X_{\mathrm{RSI}}$ is called an RSI set if, for every $t\in\mathbb N$,
$x(t)\in\mathcal X_{\mathrm{RSI}}$ and
$e(t)\in\Delta(\boldsymbol\varepsilon)$ imply
$x(t+1)
\in\mathcal X_{\mathrm{RSI}}$.
\end{definition}

Definition~\ref{def:epsilon_rsi} defines the overall RSI set by jointly
considering individual agent states and inter-agent relative-state
constraints under observer-induced perturbations. The following theorem
addresses the individual-state component by decomposing its robust safe requirement into nominal contraction and perturbation accommodation.

\begin{theorem}
\label{thm:epsilon_two_condition}
Based on Definition~\ref{def:epsilon_rsi} and
Theorem~\ref{thm:ei_bound}, for each agent $i\in\mathcal V$, the local
one-step safety condition
\begin{equation}
\bigl((\bar A_i+B_iK_i)\chi_i+e_i\bigr)^\top
Q_i
\bigl((\bar A_i+B_iK_i)\chi_i+e_i\bigr)
\le1
\label{eq:robust_condition_main}
\end{equation}
holds for all $\chi_i\in\mathcal D_i$ and
$e_i\in\Delta(\varepsilon_i)$ if there exists
$\kappa_i\in(0,1)$ such that

\textbf{a)}
$
\chi_i^\top
(\bar A_i+B_iK_i)^\top
Q_i
(\bar A_i+B_iK_i)
\chi_i
\le\kappa_i, \forall \chi_i\in\mathcal D_i$;

\textbf{b)}
$
(\zeta_i+\tilde e)^\top Q_i(\zeta_i+\tilde e)\le1
$
for all $\zeta_i\in\mathbb R^{n_x}$ satisfying
$\zeta_i^\top Q_i\zeta_i\le\kappa_i$ and all
$\tilde e\in\Delta(\varepsilon_i)$.

Here, $K_i\in\mathbb R^{n_u\times n_{\chi_i}}$ is the feedback gain in~\eqref{eq:e_feedback_law};
$Q_i\in\mathbb R^{n_x\times n_x}$ with $Q_i\succ0$ defines
$\mathcal X_{\mathrm{safe}}^i$ in
Definition~\ref{def:epsilon_rsi};
$
\bar A_i:=
\begin{bmatrix}
A_i&0_{n_x\times(n_{\chi_i}-n_x)}
\end{bmatrix};
$
$
\mathcal D_i:=
\mathcal X_{\mathrm{safe}}^i
\times
\prod_{p\in\mathcal R_i}\mathcal X_{\mathrm{safe}}^p
$
is the admissible domain of the full-information vector $\chi_i$;
and $\Delta(\varepsilon_i)$ is the admissible perturbation set in
Definition~\ref{def:epsilon_rsi}, with
$\varepsilon_i\in\mathbb R_{\ge0}$ established in
Theorem~\ref{thm:ei_bound}. 

For every $x(t)\in\mathcal X_{\mathrm{RSI}}$,
Definition~\ref{def:epsilon_rsi} implies
$\chi_i(t)\in\mathcal D_i$, while
Theorem~\ref{thm:ei_bound} guarantees
$e_i(t)\in\Delta(\varepsilon_i)$. Consequently,
$x_i(t+1)\in\mathcal X_{\mathrm{safe}}^i$ for every
$x(t)\in\mathcal X_{\mathrm{RSI}}$ and $t\in\mathbb N$.
\end{theorem}

The proof of Theorem~\ref{thm:epsilon_two_condition} is presented in the Appendix. Theorem~\ref{thm:epsilon_two_condition} provides a sufficient condition
for preserving the individual-state component of the RSI set. The following theorem establishes a tractable condition for relative-state safety.

\begin{theorem}
\label{thm:relative_forward_safety}
Under Definition~\ref{def:epsilon_rsi} and
Theorem~\ref{thm:ei_bound}, consider any agent
$i\in\mathcal V$ and any $p\in\mathcal R_i$. If $\bar b_p^i>0$ and
\begin{equation}
\left(
\Theta_p^i+d_iB_iK_i
\right)
\tilde L_i
\left(
\Theta_p^i+d_iB_iK_i
\right)^\top
\le
(\bar b_p^i)^2,
\label{eq:relative_forward_condition}
\end{equation}
then, for every $x(t)\in\mathcal X_{\mathrm{RSI}}$, every
$e_i(t)\in\Delta(\varepsilon_i)$, and every neighboring input
satisfying $\|u_p(t)\|_\infty\le u_{\max}^p$, the relative state
satisfies
$
y_p^i(t+1)
=
x_i(t+1)-x_p(t+1)
\in\mathcal S_p^i.$
Here,
$
\Theta_p^i
:=
\left[
d_iA_i\ \ 0\ \ \cdots\ \ -d_iA_p\ \ \cdots\ \ 0
\right]
\in\mathbb R^{1\times n_{\chi_i}},$
where $-d_iA_p$ occupies the block associated with $x_p$ in
$\chi_i$. Moreover,
$
\tilde L_i
:=
(m_i+1)\operatorname{blkdiag}
\bigl(L_i,L_{p_1},\ldots,L_{p_{m_i}}\bigr),
$
where
$m_i:=|\mathcal R_i|$,
$\mathcal R_i =
\mathcal N_i\cup\mathcal N_i^{k\text{-hop}}=\{p_1,\ldots,p_{m_i}\}$ is ordered consistently with
$\chi_i$, and $L_a:=Q_a^{-1}$ for every
$a\in\{i\}\cup\mathcal R_i$. The matrices $Q_a$ and the sets
$\mathcal X_{\mathrm{RSI}}$ and
$\Delta(\varepsilon_i)$ are defined in
Definition~\ref{def:epsilon_rsi};
$\mathcal S_p^i$, $d_i$, and $b_p^i$ are defined in
\eqref{eq:relative_safe_set}. The remaining relative-safety margin is
$
\bar b_p^i
:=
b_p^i
-\varepsilon_i\|d_i\|_2
-u_{\max}^{p}\|d_iB_p\|_1.$
\end{theorem}

The proof of Theorem~\ref{thm:relative_forward_safety} is presented in
the Appendix. The theorem guarantees relative-state safety under observer-induced perturbations.
Together, Theorems~\ref{thm:epsilon_two_condition} and
\ref{thm:relative_forward_safety} provide sufficient conditions for
preserving the two components of the RSI set defined in Definition~\ref{def:epsilon_rsi}.

\subsection{Joint LMI-Based Observer-Controller Synthesis}

Having established sufficient conditions for preserving both components
of the RSI set, this subsection combines them with observer stability and the input constraints in a joint LMI-based synthesis framework. 

\begin{definition}\label{prob:joint_lmi_synthesis}
Based on
Assumption~\ref{ass:observer_stability}, Theorem~\ref{thm:epsilon_two_condition}, and Theorem~\ref{thm:relative_forward_safety}, for each
$j\in\mathcal V$ with $\eta_j>0$, define the decision matrices
$
P_j:=
\operatorname{blkdiag}
\left(
P_j^{r_1},\ldots,P_j^{r_{\eta_j}}
\right)$
and
$Y_j:=
\operatorname{blkdiag}
\left(
Y_j^{r_1},\ldots,Y_j^{r_{\eta_j}}
\right)
$ with $\mathcal O_j = \{r_1,\ldots,r_{\eta_j}\}$ as defined in Definition~\ref{def:observer_information_structure}, 
where
$
Y_j=P_j\Lambda_j^o.
$
For fixed parameters
$q_j^o\in(0,1)$ and $\gamma_j\ge1$ for all
$j\in\mathcal V$ with $\eta_j>0$, and
$\kappa_i\in(0,1)$ and $\rho_i\ge1$ for all $i\in\mathcal V$,
define the joint
optimization problem
\begin{equation}
\mathrm{OP}:
\min_{\substack{
\{P_j,Y_j\}_{j\in\mathcal V,\eta_j>0}\\
\{L_i,F_i,U_i,\beta_i\}_{i\in\mathcal V},\
\{\tau_{i,\ell}\}_{i\in\mathcal V,\,
\ell=1,\ldots,n_u}}}
-\sum_{i\in\mathcal V}\log\det(L_i)
\label{eq:OP_obj}
\end{equation}
\begin{equation}
s.t. 
\begin{bmatrix}
q_j^oP_j &
\left(
P_j\mathbf A_j^o
-Y_j\bar M_j
\right)^\top\\
P_j\mathbf A_j^o
-Y_j\bar M_j &
P_j
\end{bmatrix}
\succeq0
\label{eq:OP_observer_1}
\end{equation}
\begin{equation}
I_{\eta_j n_x}
\preceq P_j
\preceq
\gamma_jI_{\eta_j n_x},
\label{eq:OP_observer_2}
\end{equation}
\begin{equation}
\begin{bmatrix}
\kappa_iL_i &
\bar A_i\tilde L_i+B_iF_i\\
\left(
\bar A_i\tilde L_i+B_iF_i
\right)^\top &
\tilde L_i
\end{bmatrix}
\succeq0,
\label{eq:OPi_1}
\end{equation}
\begin{equation}
L_i\succeq\alpha_iI_{n_x},
\label{eq:OPi_2}
\end{equation}
\begin{equation}
c_{i,a}L_ic_{i,a}^\top\le1,
\quad
a=1,\ldots,n_{c_i},
\label{eq:OPi_3}
\end{equation}
\begin{equation}
\begin{bmatrix}
(\bar b_p^i)^2 &
\Theta_p^i\tilde L_i+d_iB_iF_i\\
\left(
\Theta_p^i\tilde L_i+d_iB_iF_i
\right)^\top &
\tilde L_i
\end{bmatrix}
\succeq0,
\forall p\in\mathcal R_i,
\label{eq:OPi_relative}
\end{equation}
\begin{equation}
\begin{bmatrix}
\frac{1}{2}{u_{\max}^i}^2-\tau_{i,\ell} &
f_{i,\ell}^\top\\
f_{i,\ell} &
\tilde L_i
\end{bmatrix}
\succeq0,
\quad
\forall i\in\mathcal V,\quad
\ell=1,\ldots,n_u,
\label{eq:OPi_4}
\end{equation}
\begin{equation}
\begin{bmatrix}
\tau_{i,\ell} &
f_{i,\ell}^\top\\
f_{i,\ell} &
U_i
\end{bmatrix}
\succeq0
\label{eq:OPi_5}
\end{equation}
\begin{equation}
\delta_i^2U_i
\preceq
\underline{\alpha}_i\tilde L_i,
\label{eq:OPi_6}
\end{equation}
\begin{equation}
\tilde L_i\succeq\beta_iI_{n_{\chi_i}},
\qquad
L_i\preceq\rho_i\beta_iI_{n_x}.
\label{eq:OPi_7}
\end{equation}
Here,
$P_j\in\mathbb R^{\eta_j n_x\times\eta_j n_x}$ is symmetric positive
definite and
$Y_j\in\mathbb R^{\eta_j n_x\times\eta_j n_x}$ is the
observer decision matrix. Moreover,
$L_i=Q_i^{-1}\in\mathbb R^{n_x\times n_x}$ is symmetric positive
definite,
$F_i\in\mathbb R^{n_u\times n_{\chi_i}}$ is the controller decision
matrix,
$U_i\in\mathbb R^{n_{\chi_i}\times n_{\chi_i}}$ is symmetric positive
semidefinite, $\tau_{i,\ell}\ge0$, and $\beta_i>0$. The matrix
$\tilde L_i$ is defined in Theorem \ref{thm:relative_forward_safety}. Furthermore,
$
F_i:=K_i\tilde L_i,
$
$
f_{i,\ell}:=F_i(\ell,:)^\top,
$
$
c_{i,a}:=C_i(a,:),
$
$
\varepsilon_i:=\sqrt{\kappa_i\rho_i}\,\delta_i,
$
$
\alpha_i:=
\frac{\varepsilon_i^2}
{(1-\sqrt{\kappa_i})^2}
$ and $\underline{\alpha}_i
:=
(m_i+1)
\min_{r\in\{i\}\cup\mathcal R_i}
\alpha_r.$ After solving $\mathrm{OP}$ in \eqref{eq:OP_obj}, the observer and controller gains are
recovered as
$
\Lambda_j^o=P_j^{-1}Y_j
$
and
$
K_i=F_i\tilde L_i^{-1},
$
respectively.
\end{definition}


\begin{remark}\label{rem:epsilon_decoupling}
The observer and controller are coupled through
$e_i(t)=-B_iK_iE_ie_i^o(t)$. This coupling is separated at the bound
level: the observer analysis gives $\|e_i^o(t)\|\le\delta_i$, while
\eqref{eq:OPi_1} and \eqref{eq:OPi_7} bound the controller-dependent
gain by $\sqrt{\kappa_i\rho_i}$. Thus,
$\varepsilon_i=\sqrt{\kappa_i\rho_i}\,\delta_i$ is fixed before the
inner optimization when the outer parameters are fixed, preserving the
convexity of the joint synthesis problem.
\end{remark}


Constraints \eqref{eq:OP_observer_1}-\eqref{eq:OP_observer_2}
enforce the observer-stability requirement in
Assumption~\ref{ass:observer_stability}. Constraints
\eqref{eq:OPi_1}-\eqref{eq:OPi_2} enforce the perturbation accommodation conditions in
Theorem~\ref{thm:epsilon_two_condition}. Constraint \eqref{eq:OPi_3}
enforces containment of the RSI set in the local safe
set. Constraint \eqref{eq:OPi_relative} enforces the relative-state
safety condition in Theorem~\ref{thm:relative_forward_safety}.
Constraints \eqref{eq:OPi_4}--\eqref{eq:OPi_6} enforce the input
constraints. Constraint~\eqref{eq:OPi_7} bounds the controller-induced error
amplification in Theorem~\ref{thm:ei_bound}. The following theorem connects the feasibility of $\mathrm{OP}$ in
Definition~\ref{prob:joint_lmi_synthesis} to the safety requirement, which solves
Problem~\ref{problem}.

\begin{theorem}
\label{the:theoremLMI}
Consider the joint optimization problem $\mathrm{OP}$ in
Definition~\ref{prob:joint_lmi_synthesis}. Suppose that the
input-variation and observer-initialization conditions required in
Theorems~\ref{thm:input_error_bound} and
\ref{thm:observer_convergence_bound} hold. If there exist fixed
$q_j^o\in(0,1)$ and $\gamma_j\ge1$ for all
$j\in\mathcal V$ with $\eta_j>0$, and
$\kappa_i\in(0,1)$ and $\rho_i\ge1$ for all $i\in\mathcal V$, such that
$\bar b_p^i>0$ for all $i\in\mathcal V$ and
$p\in\mathcal R_i$, and if $\mathrm{OP}$ is feasible, then the gains
$\Lambda_j^o=P_j^{-1}Y_j$ and $K_i=F_i\tilde L_i^{-1}$ ensure that, for
every $x(0)\in\mathcal X_{\mathrm{RSI}}$ and every
$e(t)\in\Delta(\boldsymbol\varepsilon)$, the resulting closed-loop
system satisfies $x(t)\in\mathcal X_{\mathrm{RSI}}$ and
$u(t)\in\mathcal U$ for all $t\in\mathbb N$. Therefore,
$\mathcal X_{\mathrm{RSI}}$ is an RSI set for the closed-loop
multi-agent system, where $\mathcal X_{\mathrm{RSI}}$ and
$\Delta(\boldsymbol\varepsilon)$ are defined in
Definition~\ref{def:epsilon_rsi}.
\end{theorem}

The proof is presented in the Appendix. Theorem~\ref{the:theoremLMI}
shows that any feasible solution yields observer and controller gains
for which the overall RSI set is invariant and the
input constraints are satisfied. Hence, feasibility of $\mathrm{OP}$ solves Problem~\ref{problem}.



\section{SIMULATION}

\begin{table}[t]
\centering
\caption{Selected parameters and observer-induced state-perturbation bounds.}
\label{tab:global_controller_parameters}
\begin{tabular}{c c c c c c c}
\hline
Agent & $\kappa_i$ & $\rho_i$ & $\gamma_i$ & $q_i^o$
& $\delta_i$ & $\varepsilon_i$\\
\hline
1 & 0.95 & 1.00 & 1.00 & 0.14 & 0.0245 & 0.0239\\
2 & 0.68 & 1.00 & 1.00 & 0.10 & 0.4538 & 0.3742\\
3 & 0.68 & 1.00 & 1.00 & 0.10 & 0.4538 & 0.3742\\
4 & 0.95 & 1.00 & 1.00 & 0.14 & 0.0245 & 0.0239\\
5 & 0.59 & 1.00 & 1.00 & 0.10 & 0.6418 & 0.4930\\
\hline
\end{tabular}
\end{table}

In this section, a case study is provided to illustrate the proposed
framework. Consider a network of five agents indexed by
$\mathcal V=\{1,2,3,4,5\}$, where $i\in\mathcal V$ denotes the agent
index. The agents are interconnected through the communication topology
$
1-2-3-4,\ 2-5,\ 3-5.
$
For each $i \in \mathcal{V}$, define
$
d_{1,i}=
\begin{bmatrix}
\cos\theta_i&\sin\theta_i
\end{bmatrix}
$
and
$
d_{2,i}=
\begin{bmatrix}
-\sin\theta_i&\cos\theta_i
\end{bmatrix}.
$
For $i\in\{1,4\}$,
$A_i=
\begin{bmatrix}
0.1639&0.1387\\
0.3556&0.4797
\end{bmatrix},
B_i=
\begin{bmatrix}
0.1513\\
0.3529
\end{bmatrix},$ $
(\theta_i,a_i,r_i)=(30^\circ,2.4,3.4),
$
whereas for $i\in\{2,3,5\}$
$
A_i=
\begin{bmatrix}
0.4248&-0.1173\\
0.1554&0.0803
\end{bmatrix},
B_i=
\begin{bmatrix}
0.3549\\
0.1062
\end{bmatrix}, $
$
(\theta_i,a_i,r_i)=(-20^\circ,2.2,3.6)
$.
The state constraint matrix is
$
C_i=[d_{1,i}/a_i^\top \ -d_{1,i}/a_i^\top \ d_{2,i}/r_i^\top \ -d_{2,i}/r_i^\top]^\top$.
Thus, each safe set is a rotated rectangle with half-widths $a_i$ and
$r_i$. The input bound is set to $u_{\max}^i=4$ for all
$i\in\mathcal V$. The relative-state safety requirement is $
|d_{1,i}(x_i-x_p)|\le b_p^i, p\in\mathcal R_i,$ where $b_p^i=6.10$ for $i\in\{1,4\}$ and
$b_p^i=6.75$ for $i\in\{2,3,5\}$.
The observer range is set to $k=2$.
The remote-state sets are
$\mathcal N_1^{k\text{-hop}}=\{3,5\}$,
$\mathcal N_2^{k\text{-hop}}=\{4\}$,
$\mathcal N_3^{k\text{-hop}}=\{1\}$,
$\mathcal N_4^{k\text{-hop}}=\{2,5\}$, and
$\mathcal N_5^{k\text{-hop}}=\{1,4\}$.
The communication matrices are
$M_1=M_4=
\begin{bmatrix}
2&-1\\
-1&2
\end{bmatrix},
M_2=M_3=1,
M_5=I_2.$
The initial error of each maintained state estimate is set to
$[0.01\ -0.01]^\top$. The corrected input estimates are initialized such that $\|{\mathbf e}_{u,j}^{o,+}(0)\|=0$ for all $j \in \mathcal{V}$.
The input-increment bounds are
$l_1=l_4=0.50$ and
$l_2=l_3=l_5=1.00$.
The corresponding values of
$\delta_i$, $\varepsilon_i$, and $\alpha_i$ are computed before solving
the joint convex optimization problem. The problem $\mathrm{OP}$ is modeled in
YALMIP~\cite{YALMIP} and solved using MOSEK~\cite{MOSEK}.
The selected outer parameters and observer-induced state-perturbation
bounds are summarized in
Table~\ref{tab:global_controller_parameters}.

For conciseness, we report the synthesis results associated with
agent 3. The joint convex optimization problem gives
$
P_3 = I_2,
\Lambda_3^o=
\begin{bmatrix}
0.4248 & -0.1173\\ 
0.1554 & 0.0803
\end{bmatrix}.
$
The recovered controller gain and RSI matrix are $K_3=[-0.2160 \ 0.0739 \ 0.0058 \ -0.0022 \ 0.0002 \ -0.0002 \ 0.0058 \ -0.0022 \ 0.0002 \ -0.0002],$ and $L_3=
\begin{bmatrix}
5.7900 & 2.6098\\
2.6098 & 12.0100
\end{bmatrix}.$
For the closed-loop verification, 100 initial conditions are sampled
from the RSI sets while satisfying the prescribed relative-state
safety constraints, and the system is simulated for $H=100$ time
steps. 
Fig.~\ref{fig:error_bound} shows the state-estimation errors governed
by \eqref{eq:error_dynamics}.
Fig.~\ref{fig:rsi_projection} shows the local safe sets and closed-loop
trajectories of all agents. Every trajectory remains in its local safe set. Fig.~\ref{fig:control_input} presents the
control input trajectories of all agents, all of which satisfy
$\|u_i(t)\|_\infty\le4$.
Fig.~\ref{fig:relative_safety} shows the normalized relative-state
safety measures $
\max_{p\in\mathcal R_i}
\frac{|d_{1,i}(x_i(t)-x_p(t))|}{b_p^i},$
which remain below 1 throughout the simulation.
These results illustrate the overall RSI conditions and the input constraints.

\begin{figure}[t]
    \centering   \includegraphics[width=0.83\linewidth]{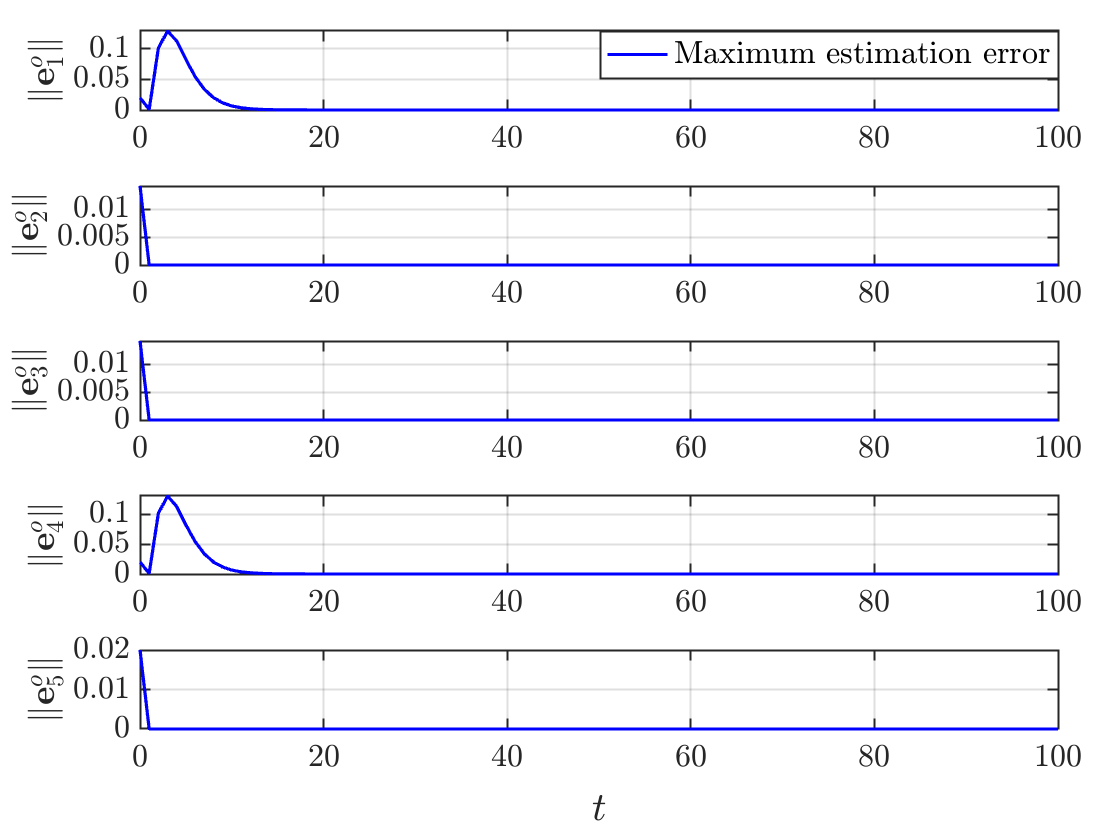}
    \caption{Observer estimation errors of all agents.}
    \label{fig:error_bound}
\end{figure}

\begin{figure}[t]
    \centering  \includegraphics[width=0.83\linewidth]{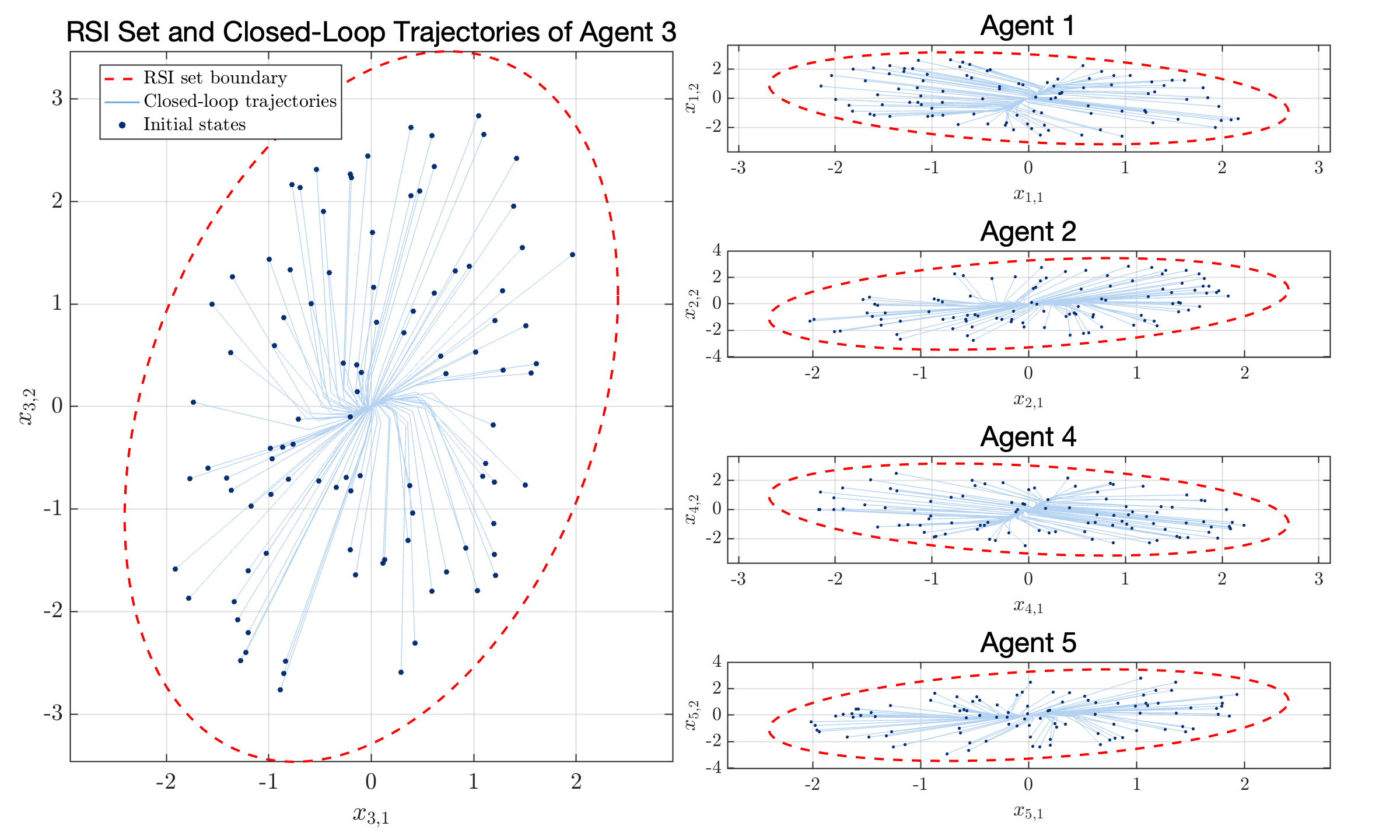}
    \caption{Local safe sets and closed-loop trajectories of all agents.}
    \label{fig:rsi_projection}
\end{figure}

\begin{figure}[t]
    \centering   \includegraphics[width=0.85\linewidth]{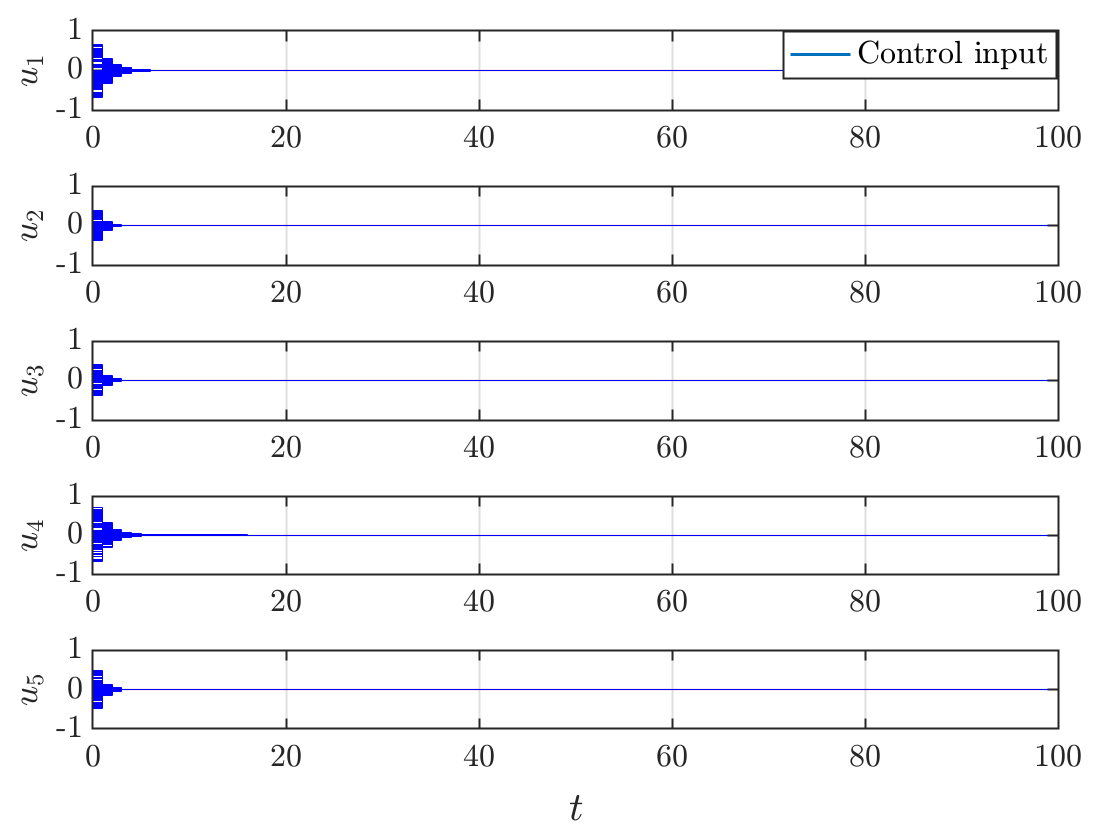}
    \caption{Control input trajectories of all agents.}
    \label{fig:control_input}
\end{figure}

\begin{figure}[t]
    \centering    \includegraphics[width=0.83\linewidth]{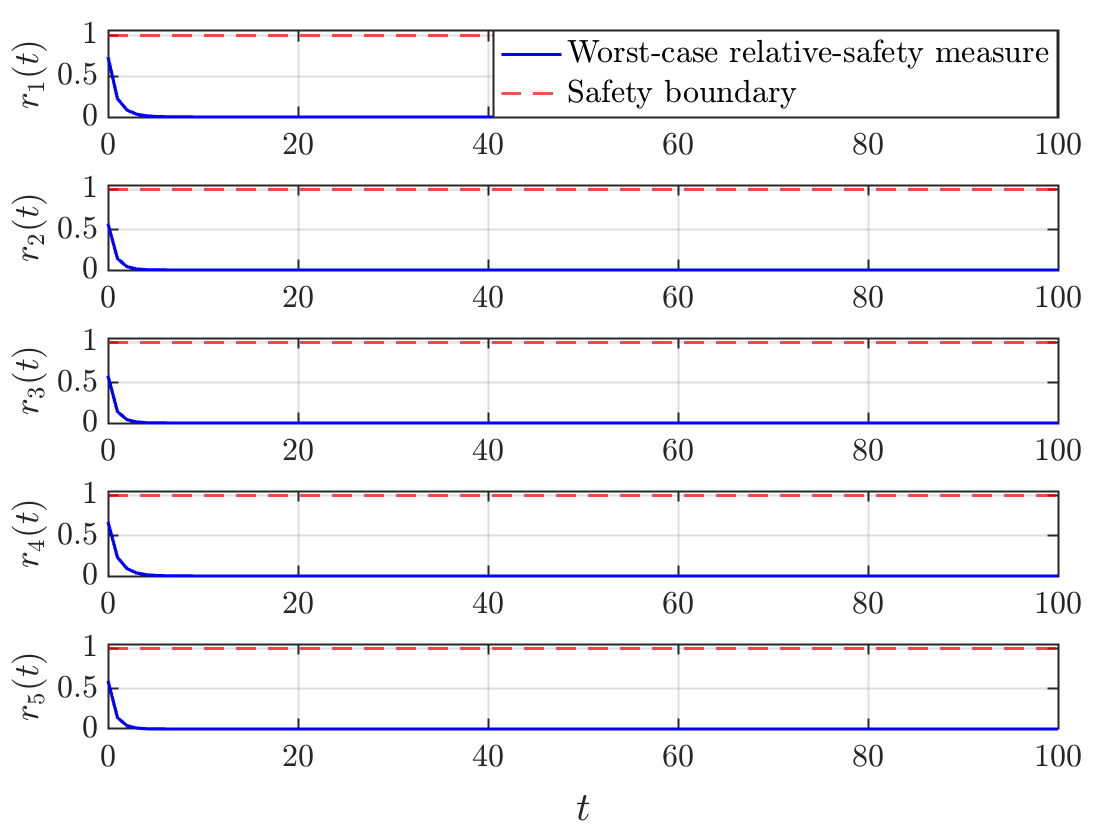}
    \caption{Normalized relative-state safety measures of all agents.}
    \label{fig:relative_safety}
\end{figure}

\section{CONCLUSION}

This paper developed a communication-aware safe control framework for
discrete-time linear multi-agent systems with limited inter-agent
information exchange. For a prescribed $k$, a distributed input
predictor and a distributed $k$-hop state observer were constructed to
reconstruct unavailable remote inputs and states using locally available
information. Uniform estimation-error bounds were derived and propagated
through the feedback controllers to characterize the resulting
observer-induced state perturbations. Based on these bounds, an overall
RSI set was formulated to jointly account for individual-state and
inter-agent relative-state safety constraints. An LMI-based convex
optimization framework was then developed to jointly synthesize the
distributed observers, state-feedback controllers, and RSI set while satisfying the input constraints. The case
study illustrates the effectiveness of the proposed
communication-aware safe control framework.




\section*{APPENDIX}

\noindent\textbf{\textit{Proof of Lemma~\ref{M_kc}.}}
Since $D_j$ is a Laplacian matrix and $H_j$ is diagonal with
nonnegative entries, both are positive semidefinite. If
$a^\top M_ja=0$, then $a^\top D_ja=0$ and $a^\top H_ja=0$.
The first equality implies that $a$ is constant on each connected
component of the subgraph induced by $\mathcal O_j$, whereas the second
implies that every component containing a positive diagonal entry of
$H_j$ has zero constant value. By
Assumption~\ref{ass:graph_condition}, every connected component contains
such a node. Hence $a=0$, so $M_j\succ0$. Since
$\bar M_j=M_j\otimes I_{n_x}$, it follows that $\bar M_j\succ0$.
\hfill$\blacksquare$

\noindent\textbf{\textit{Proof of Lemma~\ref{lem:input_predictor}.}}
For every $r\in\mathcal O_j$, agent $r$ locally applies
$\hat u_j^{r,+}(t)=\hat u_j^r(t)+s_j\xi_{u,j}^r(t)$ and
$\hat u_j^r(t+1)=\hat u_j^{r,+}(t)$, where
$\xi_{u,j}^r(t)$ is given by
\eqref{eq:input_observer_correction}. Stacking these local updates
according to the ordering
$\mathcal O_j=\{r_1,\ldots,r_{\eta_j}\}$ gives
\eqref{eq:corrected_input_estimate}. Moreover,
Lemma~\ref{M_kc} gives
$0<\lambda_{\min}(M_j)\le\lambda_{\max}(M_j)$, so
$s_j=2/(\lambda_{\min}(M_j)+\lambda_{\max}(M_j))$ is positive and
well defined. Hence, the stated input predictor is well defined and
can be implemented distributively under
Assumption~\ref{as:local_information}.
\hfill$\blacksquare$

\noindent\textbf{\textit{Proof of Theorem~\ref{thm:input_error_bound}.}}
Let $\mathbf e_{u,j}^{o}(t):=\mathbf1_{\eta_j}\otimes u_j(t)-\mathbf{\hat u}_j^o(t)$.
Stacking \eqref{eq:input_observer_correction} gives
$\xi_{u,j}^o(t)=(M_j\otimes I_{n_u})\mathbf e_{u,j}^o(t)$.
It follows from \eqref{eq:corrected_input_estimate} that $\mathbf e_{u,j}^{o,+}(t)=((I_{\eta_j}-s_jM_j)\otimes I_{n_u})\mathbf e_{u,j}^{o}(t)$.
Since $\mathbf{\hat u}_j^o(t+1)=\mathbf{\hat u}_j^{o,+}(t)$,
the corrected error satisfies $\mathbf e_{u,j}^{o,+}(t+1)=((I_{\eta_j}-s_jM_j)\otimes I_{n_u})(\mathbf e_{u,j}^{o,+}(t)+\mathbf1_{\eta_j}\otimes(u_j(t+1)-u_j(t)))$.
Since $M_j\succ0$ by Lemma~\ref{M_kc} and $s_j=2/(\lambda_{\min}(M_j)+\lambda_{\max}(M_j))$ by Lemma~\ref{lem:input_predictor}, the symmetric matrix $I_{\eta_j}-s_jM_j$ has all its eigenvalues in $(-1,1)$. Hence $\sigma_j=\|I_{\eta_j}-s_jM_j\|<1$. Taking norms and using $\|u_j(t+1)-u_j(t)\|\le l_j$ gives $\|\mathbf e_{u,j}^{o,+}(t+1)\|\le\sigma_j\|\mathbf e_{u,j}^{o,+}(t)\|+\sigma_j\sqrt{\eta_j}l_j$. Because $\nu_j\ge\bar u_j^o$ and $\nu_j\ge\sigma_j\sqrt{\eta_j}l_j/(1-\sigma_j)$, induction yields $\|\mathbf e_{u,j}^{o,+}(t)\|\le\nu_j$ for all $t\in\mathbb N$.
\hfill$\blacksquare$

\noindent\textbf{\textit{Proof of Lemma~\ref{lem:observer_error_dynamics}.}}
Stacking \eqref{eq:implementable_innovation} and using the definitions
of $D_j$ and $H_j$ gives $\xi_j^o(t)=\bar M_j\mathbf e_j^o(t)$.
Stacking $\eta_j$ copies of the subsystem dynamics gives $\mathbf1_{\eta_j}\otimes x_j(t+1)=\mathbf A_j^o(\mathbf1_{\eta_j}\otimes x_j(t))+\mathbf B_j^o(\mathbf1_{\eta_j}\otimes u_j(t))$.
Subtracting \eqref{eq:observer} and using the definitions of
$\mathbf e_j^o$ and $\mathbf e_{u,j}^{o,+}$ gives $\mathbf e_j^o(t+1)=\mathbf A_j^o\mathbf e_j^o(t)+\mathbf B_j^o\mathbf e_{u,j}^{o,+}(t)-\Lambda_j^o\xi_j^o(t)=(\mathbf A_j^o-\Lambda_j^o\bar M_j)\mathbf e_j^o(t)+\mathbf B_j^o\mathbf e_{u,j}^{o,+}(t)$, which is \eqref{eq:error_dynamics}.
\hfill$\blacksquare$

\noindent\textbf{\textit{Proof of Theorem~\ref{thm:observer_convergence_bound}.}}
Iterating the observer-error dynamics in \eqref{eq:error_dynamics}
gives $\mathbf e_j^o(t)=\Psi_j^t\mathbf e_j^o(0)+\sum_{\tau=0}^{t-1}\Psi_j^{t-1-\tau}\mathbf B_j^o\mathbf e_{u,j}^{o,+}(\tau)$.
The Lyapunov inequality \eqref{eq:lyapunov_condition} implies $\|\Psi_j^h\|\le\sqrt{\lambda_{\max}(P_j)/\lambda_{\min}(P_j)}(q_j^o)^{h/2}$ for every $h\in\mathbb N$.
Using this inequality, the initial-error condition, and the input-error
bound $\|\mathbf e_{u,j}^{o,+}(\tau)\|\le\nu_j$ from
Theorem~\ref{thm:input_error_bound} yields $\|\mathbf e_j^o(t)\|\le\sqrt{\lambda_{\max}(P_j)/\lambda_{\min}(P_j)}((q_j^o)^{t/2}\bar e_{j,0}^o+(1-(q_j^o)^{t/2})\|\mathbf B_j^o\|\nu_j/(1-\sqrt{q_j^o}))$.
By \eqref{eq:observer_P_bounds}, $\lambda_{\max}(P_j)/\lambda_{\min}(P_j)\le\gamma_j$. Since $(q_j^o)^{t/2}\le1$ and $1-(q_j^o)^{t/2}\le1$, it follows that $\|\mathbf e_j^o(t)\|\le\sqrt{\gamma_j}(\bar e_{j,0}^o+\|\mathbf B_j^o\|\nu_j/(1-\sqrt{q_j^o}))=\bar e_j^o$.
\hfill$\blacksquare$

\noindent\textbf{\textit{Proof of Theorem~\ref{thm:ei_bound}.}}
The selector matrix in the theorem is explicitly $E_i:=\operatorname{col}(0_{n_x\times\eta_i n_x},0_{|\mathcal N_i|n_x\times\eta_i n_x},I_{\eta_i n_x})\in\mathbb R^{n_{\chi_i}\times\eta_i n_x}$.
Because the locally measured components of $\chi_i(t)$ and
$\hat\chi_i(t)$ coincide, while their remote-state components differ by $e_i^o(t)=x_i^o(t)-\hat x_i^o(t)$, one has $\chi_i(t)-\hat\chi_i(t)=E_ie_i^o(t)$.
Let $\bar A_i:=\begin{bmatrix}A_i&0_{n_x\times(n_{\chi_i}-n_x)}\end{bmatrix}$.
Since the first block of $\chi_i(t)$ is $x_i(t)$,
$\bar A_i\chi_i(t)=A_ix_i(t)$. Substituting the
implemented feedback law $u_i(t)=K_i\hat\chi_i(t)$ from
\eqref{eq:e_feedback_law} into the subsystem dynamics
\eqref{eq:agent_dynamics} gives $x_i(t+1)=\bar A_i\chi_i(t)+B_iK_i\hat\chi_i(t)=(\bar A_i+B_iK_i)\chi_i(t)+B_iK_i(\hat\chi_i(t)-\chi_i(t))$.
The first term is the successor generated by the ideal controller
$u_i(t)=K_i\chi_i(t)$ with full information. The second term is the
difference between the implemented and ideal closed-loop successors; hence $e_i(t):=B_iK_i(\hat\chi_i(t)-\chi_i(t))=-B_iK_iE_ie_i^o(t)$.
Thus, $e_i(t)$ is exactly the additive closed-loop deviation induced
by the remote-state estimation error.
For each $j\in\mathcal N_i^{k\text{-hop}}$,
$i\in\mathcal O_j$ and $x_j(t)-\hat x_j^i(t)$ is one block of
$\mathbf e_j^o(t)$. Theorem~\ref{thm:observer_convergence_bound}
therefore yields $\|x_j(t)-\hat x_j^i(t)\|\le\|\mathbf e_j^o(t)\|\le\bar e_j^o$.
Since $e_i^o(t)$ stacks these errors over
$j\in\mathcal N_i^{k\text{-hop}}$, $\|e_i^o(t)\|^2=\sum_{j\in\mathcal N_i^{k\text{-hop}}}\|x_j(t)-\hat x_j^i(t)\|^2\le\sum_{j\in\mathcal N_i^{k\text{-hop}}}(\bar e_j^o)^2=\delta_i^2$.
Finally, the matrix inequality
$(B_iK_iE_i)(B_iK_iE_i)^\top\preceq g_i^2I_{n_x}$ implies
$\|B_iK_iE_i\|_2\le g_i$. Hence $\|e_i(t)\|\le\|B_iK_iE_i\|_2\|e_i^o(t)\|\le g_i\delta_i=\varepsilon_i$.
\hfill$\blacksquare$

\noindent\textbf{\textit{Proof of Theorem~\ref{thm:epsilon_two_condition}.}}
For any $\chi_i\in\mathcal D_i$, let
$\zeta_i:=(\bar A_i+B_iK_i)\chi_i$.
Condition \textbf{a)} gives
$\zeta_i^\top Q_i\zeta_i\le\kappa_i$.
Applying condition \textbf{b)} to this $\zeta_i$ and any
$e_i\in\Delta(\varepsilon_i)$ yields $((\bar A_i+B_iK_i)\chi_i+e_i)^\top Q_i((\bar A_i+B_iK_i)\chi_i+e_i)\le1$, which is \eqref{eq:robust_condition_main}. For every $x(t)\in\mathcal X_{\mathrm{RSI}}$, Definition~\ref{def:epsilon_rsi} gives $\chi_i(t)\in\mathcal D_i$, and Theorem~\ref{thm:ei_bound} gives $e_i(t)\in\Delta(\varepsilon_i)$; hence $x_i(t+1)\in\mathcal X_{\mathrm{safe}}^i$.
\hfill$\blacksquare$

\noindent\textbf{\textit{Proof of Theorem~\ref{thm:relative_forward_safety}.}}
For every $\chi_i\in\mathcal D_i$, $\chi_i^\top\tilde L_i^{-1}\chi_i=(x_i^\top L_i^{-1}x_i+\sum_{p\in\mathcal R_i}x_p^\top L_p^{-1}x_p)/(m_i+1)\le1$.
Thus, $\mathcal D_i$ is contained in the ellipsoid determined by
$\tilde L_i$. Using $y_p^i=x_i-x_p$ and the subsystem dynamics gives
$d_iy_p^i(t+1)=(\Theta_p^i+d_iB_iK_i)\chi_i(t)+d_ie_i(t)-d_iB_pu_p(t)$.
Let $a_p^i:=\Theta_p^i+d_iB_iK_i$. By the generalized Cauchy-Schwarz inequality, $|a_p^i\chi_i(t)|^2\le(a_p^i\tilde L_i(a_p^i)^\top)(\chi_i(t)^\top\tilde L_i^{-1}\chi_i(t))$. Hence \eqref{eq:relative_forward_condition} gives $|a_p^i\chi_i(t)|\le\bar b_p^i$. Moreover, $|d_ie_i(t)|\le\varepsilon_i\|d_i\|_2$ and $|d_iB_pu_p(t)|\le u_{\max}^p\|d_iB_p\|_1$. The triangle inequality and \eqref{eq:relative_remaining_margin} yield $|d_iy_p^i(t+1)|\le b_p^i$, proving $y_p^i(t+1)\in\mathcal S_p^i$.
\hfill$\blacksquare$

\noindent\textbf{\textit{Proof of Theorem~\ref{the:theoremLMI}.}}
Since $P_j$ and $Y_j$ are block diagonal,
$\Lambda_j^o=P_j^{-1}Y_j$ has the required distributed structure.
By the Schur complement\cite{Schur}, \eqref{eq:OP_observer_1} yields
$\Psi_j^\top P_j\Psi_j\preceq q_j^oP_j$; together with the bounds on
$P_j$, Assumption~\ref{ass:observer_stability} is satisfied.
Let $G_i:=\bar A_i+B_iK_i$. Using $F_i=K_i\tilde L_i$, the Schur complement of \eqref{eq:OPi_1} gives $G_i\tilde L_iG_i^\top\preceq\kappa_iL_i$. For every $\chi_i\in\mathcal D_i$, the proof of Theorem~\ref{thm:relative_forward_safety} gives $\chi_i^\top\tilde L_i^{-1}\chi_i\le1$. The preceding matrix inequality therefore implies $G_i^\top L_i^{-1}G_i\preceq\kappa_i\tilde L_i^{-1}$ and hence $(G_i\chi_i)^\top Q_i(G_i\chi_i)\le\kappa_i$, which establishes condition \textbf{a)} of Theorem~\ref{thm:epsilon_two_condition}.
Moreover, \eqref{eq:OPi_2} and the definition of $\alpha_i$ imply
$Q_i=L_i^{-1}\preceq\alpha_i^{-1}I_{n_x}$. Writing $\|z\|_{Q_i}:=\sqrt{z^\top Q_i z}$, one has $\|e_i\|_{Q_i}\le\varepsilon_i/\sqrt{\alpha_i}=1-\sqrt{\kappa_i}$ for every $e_i\in\Delta(\varepsilon_i)$. Whenever $\|\zeta_i\|_{Q_i}\le\sqrt{\kappa_i}$, the triangle inequality gives $\|\zeta_i+e_i\|_{Q_i}\le1$, establishing condition \textbf{b)}.
Hence $\mathcal X_{\mathrm{safe}}^i$ satisfies the local robust
one-step safety condition. Constraint \eqref{eq:OPi_3} further guarantees
$\mathcal X_{\mathrm{safe}}^i\subseteq\mathcal X^i$.
From \eqref{eq:OPi_7}, $\tilde L_i\succeq\beta_iI_{n_{\chi_i}}$ and $L_i\preceq\rho_i\beta_iI_{n_x}$. Hence $G_iG_i^\top\preceq\beta_i^{-1}G_i\tilde L_iG_i^\top\preceq\kappa_i\rho_iI_{n_x}$. Since $\bar A_iE_i=0$ and $E_iE_i^\top\preceq I_{n_{\chi_i}}$, it follows that $(B_iK_iE_i)(B_iK_iE_i)^\top=(G_iE_i)(G_iE_i)^\top\preceq G_iG_i^\top\preceq\kappa_i\rho_iI_{n_x}$. Theorem~\ref{thm:ei_bound} therefore gives $\|e_i(t)\|\le\varepsilon_i=\sqrt{\kappa_i\rho_i}\delta_i$.
Moreover, the Schur complement of \eqref{eq:OPi_relative} gives
\eqref{eq:relative_forward_condition}; hence
Theorem~\ref{thm:relative_forward_safety} guarantees the relative-state
constraint.
It remains to verify the input constraint. Let
$k_{i,\ell}:=K_i(\ell,:)$, so that
$f_{i,\ell}=\tilde L_i k_{i,\ell}^\top$.
The Schur complement of \eqref{eq:OPi_4} gives $k_{i,\ell}\tilde L_i k_{i,\ell}^\top\le\frac12(u_{\max}^i)^2-\tau_{i,\ell}$. Because $\chi_i^\top\tilde L_i^{-1}\chi_i\le1$, the generalized Cauchy--Schwarz inequality yields $|k_{i,\ell}\chi_i|^2\le\frac12(u_{\max}^i)^2-\tau_{i,\ell}$.
Furthermore, \eqref{eq:OPi_5} implies $f_{i,\ell}f_{i,\ell}^\top\preceq\tau_{i,\ell}U_i$, and hence $k_{i,\ell}^\top k_{i,\ell}\preceq\tau_{i,\ell}\tilde L_i^{-1}U_i\tilde L_i^{-1}$. Since $L_r\succeq\alpha_rI$ for every $r\in\{i\}\cup\mathcal R_i$, one has $\tilde L_i\succeq\underline{\alpha}_iI$. Combining this relation with \eqref{eq:OPi_6} gives $\delta_i^2\tilde L_i^{-1}U_i\tilde L_i^{-1}\preceq\underline{\alpha}_i\tilde L_i^{-1}\preceq I_{n_{\chi_i}}$. Thus, for $v_i:=E_ie_i^o$ with $\|v_i\|\le\delta_i$, one obtains $|k_{i,\ell}v_i|^2\le\tau_{i,\ell}$; if $\delta_i=0$, this conclusion follows directly from $v_i=0$.
Since $u_{i,\ell}=k_{i,\ell}\chi_i-k_{i,\ell}v_i$, $|u_{i,\ell}|^2\le2|k_{i,\ell}\chi_i|^2+2|k_{i,\ell}v_i|^2\le(u_{\max}^i)^2$.
Thus $\|u_i(t)\|_\infty\le u_{\max}^i$.
It remains to close the recursive safety argument. Suppose that
$x(t)\in\mathcal X_{\mathrm{RSI}}$. Then
$x_i(t)\in\mathcal X_{\mathrm{safe}}^i$ and
$x_p(t)\in\mathcal X_{\mathrm{safe}}^p$ for every
$i\in\mathcal V$ and $p\in\mathcal R_i$; hence
$\chi_i(t)\in\mathcal D_i$. The local robust one-step safety property,
together with $\|e_i(t)\|\le\varepsilon_i$, yields
$x_i(t+1)\in\mathcal X_{\mathrm{safe}}^i$ for all $i$. The input bound
proved above and Theorem~\ref{thm:relative_forward_safety} further give
$x_i(t+1)-x_p(t+1)\in\mathcal S_p^i$ for all admissible pairs
$(i,p)$. Therefore, $x(t+1)\in\mathcal X_{\mathrm{RSI}}$.
Since $x(0)\in\mathcal X_{\mathrm{RSI}}$ and the prescribed observer
initialization ensures the uniform error bounds from
Theorem~\ref{thm:observer_convergence_bound}, induction proves the state and input safety properties for all $t\in\mathbb N$.
\hfill$\blacksquare$

\end{document}